\documentclass{article}

\usepackage{amsthm}
\usepackage{amssymb}
\usepackage{amsmath}
\usepackage{enumerate}
\usepackage{textcomp}
\usepackage{verbatim}

\newtheorem{thm}{Theorem}
\newtheorem{lem}[thm]{Lemma}
\newtheorem{definition}{Definition}
\usepackage{graphicx}
\usepackage{grffile}
\date{}

\title{Small Strong Epsilon Nets}
\author{Pradeesha Ashok \thanks{Department of Computer Science and Automation,
				Indian Institute of Science, Bangalore, India.
				Email :\texttt{pradeesha@csa.iisc.ernet.in}}
\and Umair Azmi\thanks{Department of Computer Science and Automation,
				Indian Institute of Science, Bangalore, India.
				Email :\texttt{umair@csa.iisc.ernet.in}}
\and Sathish Govindarajan\thanks{Department of Computer Science and Automation,
				Indian Institute of Science, Bangalore, India.
				Email :\texttt{gsat@csa.iisc.ernet.in}}}
\begin{document}
\maketitle
% \begin{frontmatter}
% \title{Small Strong Epsilon Nets\tnoteref{t1}}
% \tnotetext[t1]{A preliminary version of this paper appeared in Proceedings of 22nd Canadian Conference on Computational Geometry\cite{AGK10}}
% \author{Pradeesha Ashok}
%  
% \ead{ pradeesha@csa.iisc.ernet.in}
%  \author{Umair Azmi}
% \ead{umair@csa.iisc.ernet.in}
% 
% \author{Sathish Govindarajan}
% \ead{ gsat@csa.iisc.ernet.in}
% \address{Dept. of Computer Science and Automation,\\ Indian Institute of Science, Bangalore, India}
\begin{abstract}
Let $P$ be a set of $n$ points in $\mathbb{R}^d$. A point $x$ is said to be a \emph{centerpoint} of $P$ if $x$ is contained in every convex object that contains more than $dn\over d+1$ points of $P$. We call a point $x$ a \emph{strong centerpoint} for a family of objects $\mathcal{C}$ if $x \in P$ is contained in every object $C \in \mathcal{C}$ that contains more than a constant fraction of points of $P$. A strong centerpoint does not exist even for halfspaces in $\mathbb{R}^2$. We prove that a strong centerpoint exists for axis-parallel boxes in $\mathbb{R}^d$ and give exact bounds. We then extend this to small strong $\epsilon$-nets in the plane and prove upper and lower bounds for $\epsilon_i^\mathcal{S}$ where $\mathcal{S}$ is the family of axis-parallel rectangles, halfspaces and disks. Here $\epsilon_i^\mathcal{S}$ represents the smallest real number in $[0,1]$ such that there exists an $\epsilon_i^\mathcal{S}$-net of size $i$ with respect to $\mathcal{S}$. 
\end{abstract}
% \begin{keyword}
%  centerpoint\sep $\epsilon$-nets \sep axis-parallel rectangles 
% \end{keyword}
% \end{frontmatter}

\section{Introduction}
Let $P$ be a set of $n$ points in $\mathbb{R}^d$. A point $x \in \mathbb{R}^d$ is said to be a \emph{centerpoint} of $P$ if any halfspace that contains $x$ contains at least $n\over{d+1}$ points of $P$. Equivalently, $x$ is a centerpoint if and only if $x$ is contained in every convex object that contains more than $\frac{d}{d+1} n$ points of $P$. It has been proved that a centerpoint exists for any point set $P$ and the constant $\frac{d}{d+1}$ is tight~\cite{Rad46}. The computational aspects related to centerpoint have been studied in ~\cite{CEM93,JM94,MS10}. The notion of centerpoint has found many applications in statistics, combinatorial geometry, geometric algorithms etc~\cite{MTT93, MTT97, Yao83}. 

The centerpoint question has also been studied for certain special classes of convex objects. \cite{AAH09} shows exact constants on centerpoints for halfspaces, axis-parallel rectangles and disks in $\mathbb{R}^2$.
% Another related question which has been studied is regarding the existence of a centerpoint for geometric objects other than convex objects. \cite{AAH09} investigates the existence of centerpoint for objects like halfspaces, axis-parallel rectangles and discs and prove tight bounds.

By the definition of centerpoint, $x$ need not be a point in $P$.  A natural question to ask is the following: Does there exist a \emph{strong centerpoint} i.e., the centerpoint must belong to $P$. In other words, for all point sets $P$, does there exist a point $p\in P$ such that $p$ is contained in every convex object that contains more than a constant fraction of points of $P$? It can be clearly seen that a strong centerpoint does not exist for convex objects by considering a point set with points in convex position. By the same example, a strong centerpoint does not exist even for objects like disks and halfspaces. The notion of strong centerpoints has been studied with respect to wedges in \cite{MPS09} where it was shown that there always exists a point $p \in P$ such that any ${3\pi\over2}$-wedge anchored at $p$ contains at least ${n\over d+1}$ points. We show that a strong centerpoint exists for axis-parallel boxes in $\mathbb{R}^d$ i.e. for any point set $P$ with $n$ points, there exists a point $x \in P$ such that $x$ is contained in every axis-parallel box that contains more than $\frac{2d-1}{2d}n$ points of $P$. A natural extension of strong centerpoint will be to allow more number of points and see what the bounds are. This question is related to a well studied area called $\epsilon$-nets. First we shall define $\epsilon$-nets.
\begin{definition}
Let $P$ be a set of $n$ points in $\mathbb{R}^d$ and $\mathcal{R}$ be a family of geometric objects. $N\subset P$ is called a (strong) $\epsilon$-net of $P$ with respect to $\mathcal{R}$ if $N\cap R \ne \emptyset$ for all subsets $R \in \mathcal{R}$ that has more than $\epsilon n$ points of $P$, that is, $|R\cap P| > \epsilon n$. Moreover, $N$ is called a weak $\epsilon$-net if $N\subset \mathbb{R}^d$, i.e. $N$ need not be a subset of $P$. 
\end{definition}
The concept of $\epsilon$-nets was introduced by Haussler and Welzl~\cite{HW87} and has found many applications in computational geometry, approximation algorithms, learning theory etc. It has been proved that for geometric objects with finite VC dimension $d$, there exist $\epsilon$-nets of size $O({d\over{\epsilon}} \log {{1 \over \epsilon}})$~\cite{HW87}. It has been proved that $\epsilon$-nets of size $O({1\over \epsilon})$ exist for halfspaces in $\mathbb{R}^2$ and  $\mathbb{R}^3$ and pseudo-disks in $\mathbb{R}^2$~\cite{KPW92, MSW90,PR08}. Alon~\cite{Alo12} has shown a slightly super-linear lowerbound for $\epsilon$-nets with respect to lines in $\mathbb{R}^2$ thereby disproving the conjecture that there exists linear-sized $\epsilon$-nets for families of geometric objects. Recently, it has been shown that $\epsilon$-nets of size $\Omega({1\over\epsilon} \log{\log {1\over\epsilon}})$ is needed for the family of axis-parallel rectangles in $\mathbb{R}^2$~\cite{PT11}. This bound is also tight by a previous result by Aronov \emph{et al.}~\cite{AES10}. 

The dual version of $\epsilon$-net is also well studied. Here, the input is a collection of $n$ geometric objects and we have to find a sub-collection of objects that cover all points that are contained in more than $\epsilon n$ objects. This question is extensively studied in \cite{AES10,CV07, Var09}. Recently some tight lower bounds for dual $\epsilon$-nets are proved in \cite{PT11}.

Strong $\epsilon$-nets of size independent of $n$ do not exist for convex objects since they have infinite VC dimension. However, it is shown that weak $\epsilon$-nets of size polynomial in $\frac{1}{\epsilon}$ exist for convex objects in $\mathbb{R}^d$~\cite{CEG95,Mat04, MR08}. Also, a lower bound of $\Omega(\frac{1}{\epsilon}\log^{d-1} \frac{1}{\epsilon})$ is proved in \cite{BMN11}.

Weak $\epsilon$-nets are an extension of centerpoint. A centerpoint is precisely a weak $\frac{d}{d+1}$-net of size one. Small weak $\epsilon$-nets have been studied for convex objects, disks, axis-parallel rectangles and halfspaces in ~\cite{AAH09, BZ06,Dul06, MR09}. Here the size of the weak $\epsilon$-net is fixed as a small integer $i$ and the value of $\epsilon_i$ is bounded.

In this paper, we initiate the study of small strong $\epsilon$-nets. Let $\mathcal{S}$ be a family of geometric objects. Let $\epsilon_i^\mathcal{S} \in [0,1]$ represent the smallest real number such that, for any set of points $P$, there exists a set $Q \subset P$ of size $i$ which is an $\epsilon_i^\mathcal{S}$-net with respect to $\mathcal{S}$. Thus  a strong centerpoint will be an $\epsilon_1^\mathcal{S}$-net. We investigate bounds on $\epsilon_i^\mathcal{S}$ for small values of $i$ where $\mathcal{S}$ is the family of axis-parallel rectangles, halfspaces and disks.

\subsection{Our Results}
Let $\mathcal{R,H,D}$ be the family of axis-parallel rectangles, halfspaces and disks respectively.
\begin{enumerate}
\item{For axis-parallel boxes in $\mathbb{R}^d$, we show that strong centerpoint exists and obtain tight bound of $\epsilon_1^\mathcal{R} =\frac{2d-1}{2d}$. Note that strong centerpoint does not exist for halfspaces, disks or convex objects.}

\item{We give upper and lower bounds for $\epsilon_i^{\mathcal{R}}$ in the plane for small values of $i$. We also prove some general upper  bounds for $\epsilon_i^{\mathcal{R}}$ which works for all values of $i$.}
\item{We give bounds for $\epsilon_i^{\mathcal{H}}$ which are tight for even values of $i$ and almost tight for odd values of $i$. These bounds are based on a result from~\cite{KPW92}.}
\item{For the family of disks, we give a non-trivial upper bound for $\epsilon_2^{\mathcal{D}}$. }
\end{enumerate}

In section \ref{genlbcon}, we give a general lower bound construction which will be used in subsequent sections. Section \ref{rect} discusses small strong $\epsilon$-nets for axis-parallel rectangles. Sections \ref{hp} and \ref{disk} give bounds for halfspaces and disks respectively. In section \ref{lbweak}, we prove some general lower bounds for small weak $\epsilon$-nets.

\section{General lower bound construction} \label{genlbcon}

%%%%%%%%%%%%%%%%%%%%%%%%%%%%%%%%%%%%%%%%%%%%%%%%%%%%%%% General Epsilon (j+k) lower bound
In this section, we give a recursive construction to obtain lower bounds for $\epsilon_i^{\mathcal{S}}$ where $\mathcal{S}$ is a family of compact convex objects in $\mathbb{R}^d$. More precisely, we give a lower bound construction for $\epsilon_{j+k}^\mathcal{S}$ based on the lower bound constructions of $\epsilon_j^\mathcal{S}$ and $\epsilon_k^\mathcal{S}$. These lower bounds work for weak $\epsilon$-nets as well.

\begin{thm}\label{genlower}
 $\epsilon_{j+k}^{\mathcal{S}} \ge \frac{\epsilon_j^{\mathcal{S}} \epsilon_k^{\mathcal{S}}}{\epsilon_j^{\mathcal{S}} + \epsilon_k^{\mathcal{S}}}$ for $j, k \ge 1$
\end{thm}

\begin{proof}
 Let $P$ be a set of $n$ points in $\mathbb{R}^d$, arranged as two subsets, $P_1$ and $P_2$, containing $\frac{\epsilon_k^{\mathcal{S}}}{\epsilon_j^{\mathcal{S}} + \epsilon_k^{\mathcal{S}}} n$ and $\frac{\epsilon_j^{\mathcal{S}}}{\epsilon_j^{\mathcal{S}} + \epsilon_k^{\mathcal{S}}} n$ points respectively. Let $P_1$ be arranged corresponding to the lower bound construction for $\epsilon_j^{\mathcal{S}}$ and $P_2$ be arranged corresponding to the lower bound construction for $\epsilon_k^{\mathcal{S}}$. These two subsets are placed sufficiently far from each other. Therefore, if $N_1$ is an $\epsilon_j^{\mathcal{S}}$-net for $P_1$, there exists some $S \in \mathcal{S}$ avoiding $N_1$ such that $\vert S \cap P_1\vert = \epsilon_j^{\mathcal{S}} \vert P_1 \vert$. Since $S$ is compact and $P_2$ is placed sufficiently far from $P_1$, $S$ does not contain any point of $P_2$. Similarly, if $N_2$ is an $\epsilon_k^{\mathcal{S}}$-net for $P_2$, there exists some $S \in \mathcal{S}$ avoiding $N_2$ such that $\vert S \cap P_2\vert = \epsilon_k^{\mathcal{S}} \vert P_2 \vert$ and $S \cap P_1 = \emptyset$. Let $N$ be any $\epsilon_{j+k}^{\mathcal{S}}$-net. Then $N$ contains either $\le j$ points of $P_1$ or $\le k$ points of $P_2$. Therefore there always exists some $S \in \mathcal{S}$ that avoids $N$ and contains $\frac{\epsilon_j^{\mathcal{S}} \epsilon_k^{\mathcal{S}}}{\epsilon_j^{\mathcal{S }} + \epsilon_k^{\mathcal{S}}} n$ points.
\end{proof}

Based on this theorem, we give improved lower bounds on $\epsilon_i$ for axis-parallel rectangles in section~\ref{rect} and for small weak $\epsilon$-nets with respect to convex objects and disks in section~\ref{lbweak}.

\section{Axis-parallel Rectangles} \label{rect}
\noindent In this section, we show bounds on $\epsilon_i^{\mathcal{R}}$.

 Let $P$ be a set of $n$ points in $\mathbb{R}^d$. Assume that all points in $P$
have distinct co-ordinates, i.e. if $p=(p_1,p_2,...,p_d)$ and $q=(q_1,q_2,....,q_d)$ are two points in $P$, then $p_i \neq q_i$ for all $i$, $1\leq i \leq d$. This assumption can be easily removed by slightly perturbing the input point set such that the co-ordinates are distinct. It can be seen that an $\epsilon$-net for the
perturbed set acts as an $\epsilon$-net for the original set also(see section 4
in \cite{AAH09}). 

\subsection{Strong centerpoints in $\mathbb{R}^d$}
\noindent Let $\mathcal{R}$ be the family of axis-parallel boxes in $\mathbb{R}^d$.
\begin{thm}\label{ddim}
 $\epsilon_1^{\mathcal{R}} = \frac{2d-1}{2d}$
\end{thm}
\begin{figure}
\begin{center}
\scalebox{0.60}{\input{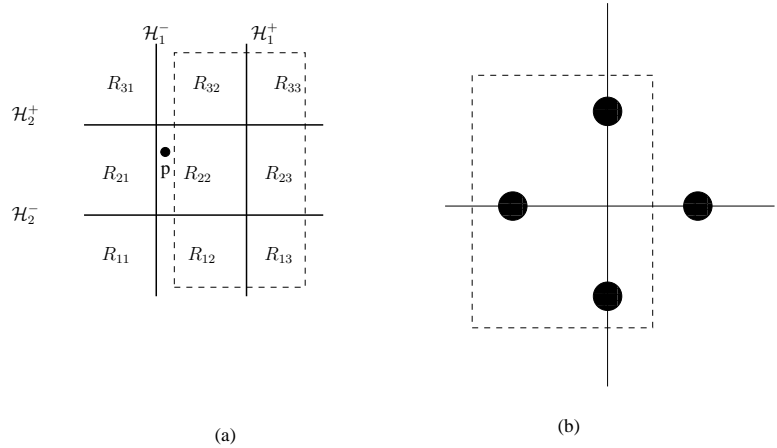}}
\caption{Bounds for $\epsilon_1^\mathcal{R}$ in $\mathbb{R}^2$}

\label{rect1}
\end{center}
\end{figure}
\begin{proof}
Let $\mathcal{H}_i^+$ and $\mathcal{H}_i^-$, $1\leq i \leq d$, be two axis-parallel hyperplanes orthogonal to the $i^{th}$ dimension that divide $P$ into three slabs. Let $\mathcal{P}_i^+ $ be the subset of $P$ contained in the positive hyperspace defined by   $\mathcal{H}_i^+$ and $\mathcal{P}_i^- $ be the subset of $P$ contained in the negative hyperspace defined by   $\mathcal{H}_i^-$. $\mathcal{H}_i^+$ and $\mathcal{H}_i^-$ are placed such that $|\mathcal{P}_i^+| = |\mathcal{P}_i^-| = \frac{n}{2d}-1$. The hyperplanes $\mathcal{H}_i^+$ and $\mathcal{H}_i^-$, $1\leq i \leq d$, partition $\mathbb{R}^d$ into $3^d$
axis-parallel $d$-dimensional boxes. Indexing the partition along each
dimension, these boxes are denoted as $R_{x_1 x_2 ... x_d}$,
where $x_i \in \{1,2,3\}$ (see figure~\ref{rect1}(a) for the upper bound construction in $\mathbb{R}^2$). Let $P_{x_1 x_2 ... x_d} =R_{x_1 x_2 ... x_d} \cap P$. We claim that $P_{22...2} \ne \emptyset$. 

Let $\mathcal{K} = \sum\limits_{i=1}^d (|\mathcal{P}_i^-| + |\mathcal{P}_i^+|) =
n-2d$. Since none of the points in  $P_{22...2}$ is counted in any
$\mathcal{P}_i^+$ or $\mathcal{P}_i^-$, $\mathcal{K} \ge n - |P_{22...2}|$. This implies that
$|P_{22...2}| \ge 2d$. 

Let $p$ be any point in $P_{22...2}$. We claim that
$\{p\}$ is a $\frac{2d-1}{2d}$-net. Any $d$-dimensional
box that does not contain $p$ has to avoid some $\mathcal{P}_i(\mathcal{P}_i^+$ or $\mathcal{P}_i^-)$ containing
$\frac{n}{2d}-1$ points. Hence it contains at most $\frac{2d-1}{2d} n$ points.

For the lower bound, place $2d $ subsets of $\frac{n}{2d}$ points such that each axis has two subsets at unit distance on either side of the origin. The lower bound construction for $\epsilon_1^\mathcal{R}$ in $\mathbb{R}^2$ is shown in figure~\ref{rect1}(b). Let $\{q\}$ be an $\epsilon$-net. Without loss of generality, assume that $q$ is chosen from the subset placed at coordinates $(1,0,0,...0)$.
Now the $d$-dimensional axis-parallel box defined by $x \le 0.5$ avoids
$q$ but contains all the
remaining $2d-1$ subsets thereby containing $\frac{2d-1}{2d}n$ points.

\end{proof}

\noindent

\subsection{Upper Bounds on $\epsilon_i^{\mathcal{R}}$}
\noindent Let $P$ be a set of $n$ points and $\mathcal{R}$ be the family of axis-parallel rectangles in $\mathbb{R}^2$. We prove upper bounds for $\epsilon_i^{\mathcal{R}}$.
\begin{lem}\label{halfhalf}
There exists a point $p\in P$ with coordinates
$(x\textquotesingle,y\textquotesingle)$ such that the
halfspaces $x \geq x\textquotesingle$ and $y \geq y\textquotesingle $ contain at least $\frac{n}{2}$ points of $P$.
\end{lem}

\begin{proof} Divide P into two horizontal and two vertical slabs such that each slab contains $\frac{n}{2}$ points (see figure~\ref{fig1}).

\begin{itemize}
 \item Case 1: $ D\cap P \neq \emptyset$. Any point $p\in D\cap P$ has the desired
property.
\item Case 2: $ D\cap P  = \emptyset$. In this case, $|A\cap P|=|C\cap
P|=\frac{n}{2}$ and $B\cap P=\emptyset$. Then the point $p \in A\cap P$ with the
smallest perpendicular distance to the horizontal line has the desired property.
\end{itemize}

\begin{figure}
\begin{center}
\includegraphics[scale=0.35]{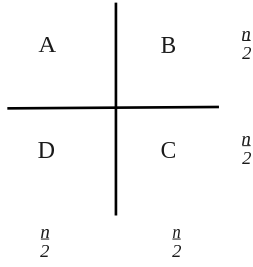}
\caption{Bisection of $P$ by vertical and horizontal lines}
\label{fig1}
\end{center}
\end{figure}
\end{proof}
\vspace{0.2in}
\begin{lem}\label{lemma2RU}
$\epsilon_{2}^{\mathcal{R}} \leq {5\over8}$	%%%%%%%%%%%%%%%%%%%%EPSILON 2 <

\end{lem}

\begin{figure}
\begin{center}
 \input{r2ub.pstex_t}
\caption{Upper bound for $\epsilon^{\cal{R}}_2$}
\label{fig2RU}
\end{center}
\end{figure}

\begin{proof}
 Divide $P$ into three horizontal slabs containing $\frac{3n}{8}$, $\frac{n}{4}$
and $\frac{3n}{8}$ points respectively. Similarly, divide $P$ into three
vertical slabs in the same proportion to get a grid with nine rectangular
regions (figure~\ref{fig2RU}(a)).
\begin{itemize}
 \item Case 1: $E \cap P \ne \emptyset$: Let $x$ be any point in E. Now $\{x\}$
is a $\frac{5}{8}$-net since any axis-parallel rectangle that avoids ${x}$ will
avoid an extreme slab which contains $\frac{3}{8}n$ points.
\item Case 2: $E \cap P = \emptyset$: Let $R_1, R_2, R_3, R_4$ be the regions $A \cup B \cup D \cup E,B\cup C \cup E\cup F, D\cup E\cup G\cup H, E\cup F\cup H\cup I$ respectively and let $P_i$ denote $R_i \cap P$ for all $i$, $1\leq i \leq 4$. Since $  E \cap P = \emptyset$, we have
$\vert P_1 \vert + \vert P_2 \vert +\vert P_3 \vert + \vert P_4 \vert = n + | (B \cup D \cup F \cup H) \cap P | =
\frac{3n}{2}$. Therefore either $(\vert P_1 \vert + \vert P_4 \vert) \geq \frac{3n}{4}$ or
$(\vert P_2 \vert + \vert P_3 \vert) \geq \frac{3n}{4}$. Without loss of generality
assume that
$(\vert P_2 \vert + \vert P_3 \vert) \geq \frac{3n}{4}$. Using lemma~\ref{halfhalf}, choose a point $p=(p_x,p_y) \in P_2$ such that the halfspaces $x \ge p_x$ and $y \ge p_y$ contain at least half of the points in $P_2$. Similarly, choose a point $q=(q_x,q_y) \in P_3$ such that the halfspaces
$x \leq q_x$ and $y \leq q_y$ contain at least half of the points in $P_3$. We claim that $\{p,q\}$ is a $\frac{5}{8}$-net. Any axis-parallel rectangle that does not take points from all the three rows and columns contains at most $\frac{5n}{8} $ points of $P$. So assume that $R$ is an axis-parallel rectangle that takes points from all the rows and columns. To avoid $\{p,q\}$, $R$ must avoid at least half of
the points from
$P_2$ as well as $P_3$
(figure~\ref{fig2RU}(b)). So it must avoid at least
$\frac{\vert P_2 \vert + \vert P_3 \vert}{2}=\frac{3n}{8}$ points. Therefore any axis-parallel rectangle
that avoids $\{p,q\}$ contains
at most $\frac{5n}{8}$ points.
\end{itemize}

\end{proof}

\vspace{0.5in}
\noindent Now we discuss two general recursive constructions for
$\epsilon_i^\mathcal{R}$. 
\begin{figure}
\begin{center}
\scalebox{0.5}{\input{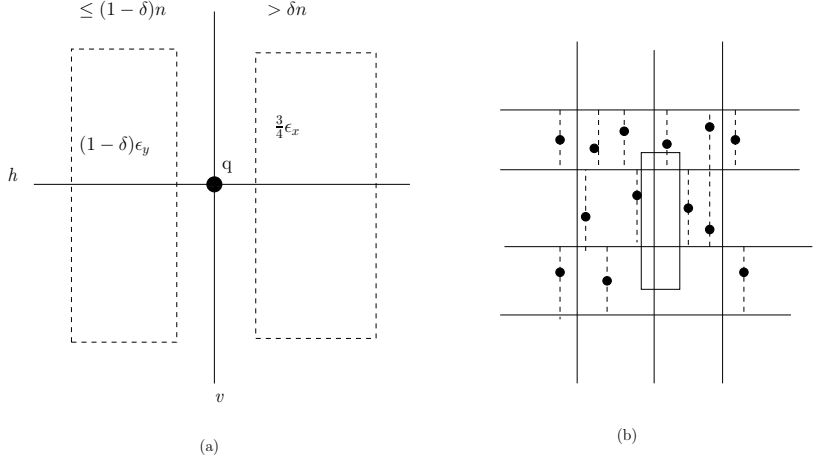}}
\caption{General constructions for axis-parallel rectangles}
\label{genfig}
\end{center}
\end{figure}
\begin{thm}\label{onept}
$\epsilon_{2(x+y)+1}^\mathcal{R} \leq
\max({{3\over{4}}\epsilon_x^\mathcal{R}},{{\epsilon_y^\mathcal{R}
\epsilon_z^\mathcal{R}}\over{\epsilon_y^\mathcal{R}+\epsilon_z^\mathcal{R}}})$
for $x,y \geq 0, x\geq y$ and $z=\lfloor{x+y\over2}\rfloor $
\end{thm}
\begin{proof}
Construct an $\epsilon_1^\mathcal{R}$-net $\{q\}$ for $P$ as described in 
Theorem~\ref{ddim}. Let $v$ and $h$ be vertical and horizontal lines through $q$ that divide $P$ into two vertical and two horizontal slabs respectively. Let $\delta \in [0,1]$ be a parameter that will be fixed later. If either of the two vertical slabs contains at least $\delta n$ points of
$P$ then construct an $\epsilon_x^\mathcal{R}$-net for the points in the slab containing at least $\delta n$ points
and $\epsilon_y^\mathcal{R} $-net for the points in the other slab. If both the vertical 
slabs contain less than $\delta n$ points, then construct an
$\epsilon_z^\mathcal{R}$-net for the points in each of the two vertical slabs. Repeat the same
construction for the horizontal slabs also. We have thus added at most
$2(x+y)+1$ points to our $\epsilon$-net $Q$.

Since $q\in Q$, any
axis-parallel rectangle that avoids $Q$ is contained in one of the 
(vertical or horizontal) slabs and this slab has at most ${3n\over4}$ points.

First consider the case where there is a vertical or horizontal slab with at least $\delta n$ points.
 After adding an $\epsilon_x^\mathcal{R}$-net to
$Q$,  any axis-parallel rectangle that avoids $Q$ contains at most
${3\over4}\epsilon_x^\mathcal{R}n$ points of $P$ from this slab. Similarly,
the other slab has at most $(1-\delta) n$ points of $P$ and any
axis-parallel rectangle that avoids $Q$ contains at most
$(1-\delta)\epsilon_y^\mathcal{R}n$ points of $P$ from this slab. Thus an
axis-parallel rectangle that avoids $Q$ and is contained in  one of these slabs
has at most $\max({3\over4}\epsilon_x^\mathcal{R}n,
(1-\delta)\epsilon_y^\mathcal{R}n )$ points(see figure~\ref{genfig} (a)). In the
case where both the slabs have less than $\delta n$ points, any
axis-parallel rectangle that avoids $Q$ has at most $\delta
\epsilon_z^\mathcal{R}n$ points. Thus any axis-parallel rectangle that avoids $Q$ has at most $max({3\over4}\epsilon_x^\mathcal{R}n, (1-\delta)\epsilon_y^\mathcal{R}n,
\delta \epsilon_z^\mathcal{R}n)$ points. Setting $\delta=
{{\epsilon_y^\mathcal{R}}\over{\epsilon_y^\mathcal{R}
+\epsilon_z^\mathcal{R}}}$ so that $(1-\delta)\epsilon_y^\mathcal{R}=\delta \epsilon_z^\mathcal{R}$, we get $\epsilon_{2(x+y)+1}^\mathcal{R} \leq
\max({{3\over{4}}\epsilon_x^\mathcal{R}},{{\epsilon_y^\mathcal{R}
\epsilon_z^\mathcal{R}}\over{\epsilon_y^\mathcal{R}+\epsilon_z^\mathcal{R}}})$.
\end{proof}
\begin{thm}\label{grid}
$\epsilon_{2(x-2)(y-1)+jx+ky}^\mathcal{R} \leq
\max({{{2\epsilon_j^{\mathcal{R}}} \over{x}},{{\epsilon_{k}^{\mathcal{R}}} \over
{y}}})$ for $x,y \geq 2$ and $j,k\geq 0$.
\end{thm}
\begin{proof}
Divide $P$ into $x$ horizontal slabs and $y$ vertical slabs to get a grid with each horizontal slab containing $n\over x$ points and each vertical slab containing
$n\over y$ points. Let the horizontal slabs be denoted as $H_1,H_2,\dots,H_x$. For a horizontal slab $H_i$, there are $y-1$ vertical lines of the grid intersecting it. For each of these lines, find two points(if present) of $P$ in $H_i$ that has the
least perpendicular distance from that line on either side. Repeat this for all horizontal slabs $H_i, 2\leq i \leq x-1$ (see figure~\ref{genfig}~(b)). Add these (at most)
$2(x-2)(y-1)$ points  to the $\epsilon^\mathcal{R}$-net $Q$.

 We claim that any
axis-parallel rectangle that avoids these points has at most $max({2n\over
x},{n\over y})$ points. If an axis-parallel rectangle intersects at most one vertical slab or at most two horizontal slabs, then it contains at most $max({2n\over
x},{n\over y})$ points. Let
$R$ be an axis-parallel rectangle that intersects at least two vertical slabs and at least three horizontal slabs. Let $H_i,H_{i+1},\dots,H_m$ be the horizontal slabs intersected by $R$. 
% Let $H_1,H_2,H_3$ be three consecutive horizontal slabs that $R$ intersects.
$R$ also intersects at least one 
vertical line and avoids the nearest points to these vertical lines in all $H_l$, $i+1 \leq l \leq m-1$. Therefore $R$ cannot take  points from any such $H_l$. Thus $R$ can only take points from $H_i$ and $H_m$ and hence contains at most ${2n\over
x}$ points of $P$.  
\vspace{0.1in}

Now add an $\epsilon_j^\mathcal{R}$-net for points in each horizontal slab and an
$\epsilon_k^\mathcal{R}$-net for points in each vertical slab and add these $\epsilon$-net points
to $Q$. Now $|Q|\leq 2(x-2)(y-1)+jx+ky$ and the result follows.

\end{proof} 

\begin{lem}
$\epsilon_3^\mathcal{R}\leq {9\over16}$; $\epsilon_5^\mathcal{R}\leq {15\over32}$;
$\epsilon_7^\mathcal{R}\leq{3\over 7}$; $\epsilon_9^\mathcal{R}\leq{5\over13}$
\end{lem}
\begin{proof}
The results follow from the fact that $\epsilon_0^\mathcal{R}=1$ and Theorem~\ref{onept} with $x=1, y=0$ for
$\epsilon_3^\mathcal{R}$; $x=2, y=0$ for
$\epsilon_5^\mathcal{R}$; $x=3, y=0$ for $\epsilon_7^\mathcal{R}$; $x=4,y=0$ for
$\epsilon_9^\mathcal{R}$.

\end{proof}
\begin{lem}
$\epsilon_4^\mathcal{R}\leq{1\over2}$; $\epsilon_8^\mathcal{R}\leq {2\over5}$;
$\epsilon_{10}^\mathcal{R}\leq{3\over8}$
\end{lem}
\begin{proof}
The results follow from the fact that $\epsilon_0^\mathcal{R}=1$ and Theorem~\ref{grid} with $x=4, y=2,j=k=0$ for
$\epsilon_4^\mathcal{R}$; with $x=5,y=2,j=0,k=1$ for $\epsilon_8^\mathcal{R}$;
$x=4,y=2,j=1,k=1$ for $\epsilon_{10}^\mathcal{R}$.

\end{proof}
\subsection{Lower Bounds on $\epsilon_i^{\mathcal{R}}$}
In this subsection, we call an axis-parallel rectangle $R$ an $\alpha$-big rectangle if $\vert P \cap R \vert \geq \alpha \vert P \vert$. Let $Q$ be an $\epsilon$-net and $P_1 \subset P$. We call $P_1$  as free if $P_1 \cap Q = \emptyset$.
\begin{lem}\label{lemma2RL}
 $\epsilon^{\cal{R}}_2 \geq \frac{5}{9}$
\end{lem}
\begin{figure}
\begin{center}
 \includegraphics[scale=0.5]{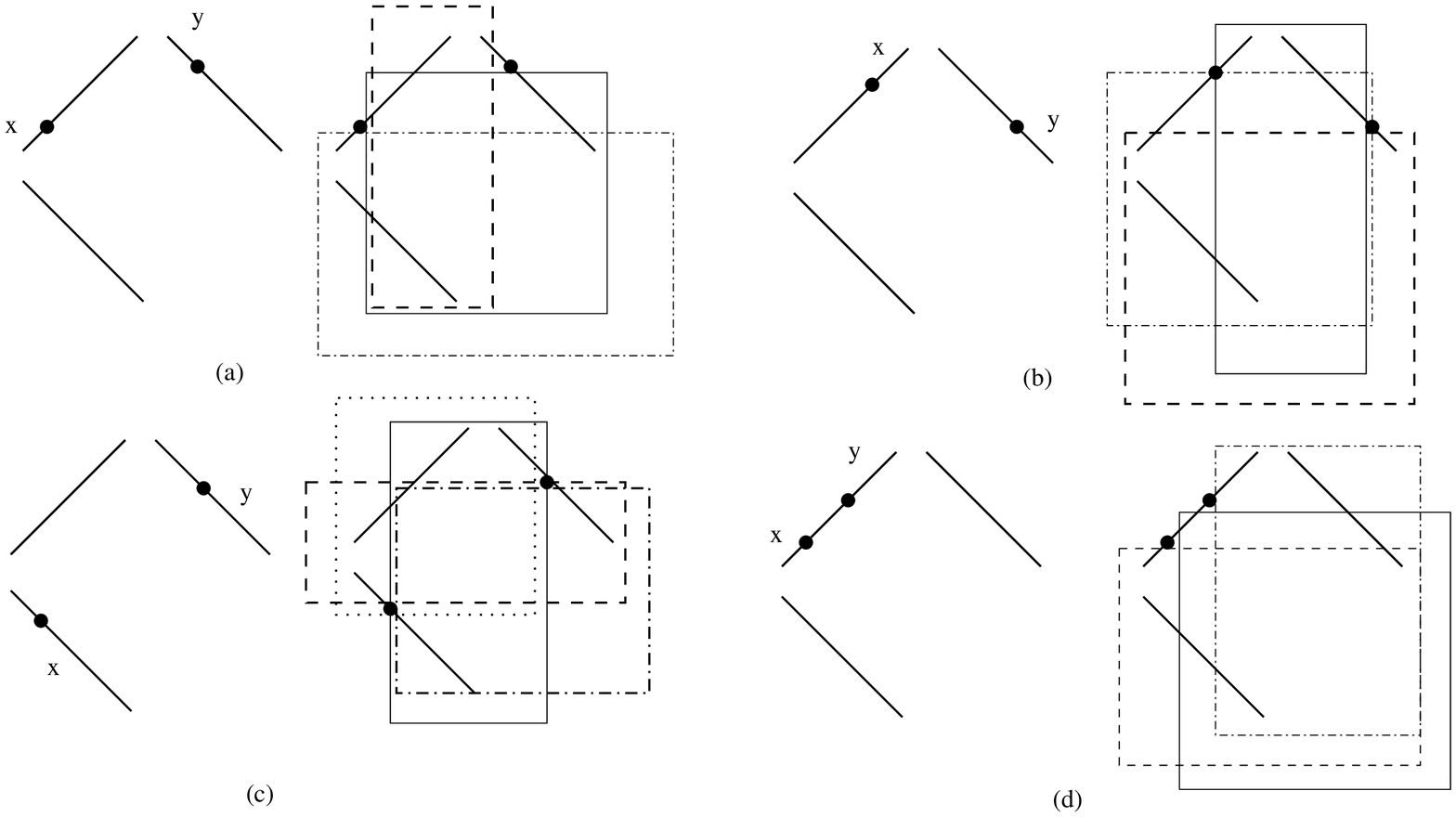}
\caption{Lower bound for $\epsilon^{\cal{R}}_2$}
\label{fig2RL}
\end{center}
\end{figure}
\begin{proof}
 Divide the $n$ points into three equal subsets of $\frac{n}{3}$
points each and place each subset uniformly
on the bold segments as shown in figure~\ref{fig2RL}. Let $Q$ be an
$\epsilon$-net of size two. The different cases of choosing
$Q$ from this set are shown in figure~\ref{fig2RL}. Let $x, y$ denote the fraction of
points from a subset which lie on one side of the point in $Q$ as shown in the
corresponding figure for each case. Let $f$ be a function that represents the
maximum fraction of points of $P$ 
that an axis-parallel rectangle contains without
containing any of the points in $Q$. In each case, we compute the values of 
$x$ and $ y $ that minimize $f$ by considering all the rectangles that avoids the $\epsilon$-net points.

\begin{enumerate}[(a)]
 \item $f=max(\frac{2x}{3}+\frac{1}{3}, 1-\frac{2(x+y)}{3},
\frac{2(1-x)}{3}+\frac{y}{3})$. $f$ is minimized when
$x=y=\frac{1}{3}$, which results in $f=\frac{5}{9}$.
 \item $f=max(\frac{2y}{3}+\frac{1}{3}, 1-\frac{2x+y}{3},
\frac{2x}{3}+\frac{1-y}{3})$. $f$ is minimized when
$x=\frac{1}{2}, y=\frac{1}{3}$, which results in $f=\frac{5}{9}$.
 \item $f=max(1-\frac{x+y}{3}, \frac{2x}{3}+\frac{1-y}{3},
\frac{2y}{3}+\frac{1-x}{3},\frac{2(x+y)-1}{3})$. $f$ is minimized when
$x=y=\frac{2}{3}$, which results in $f=\frac{5}{9}$.
 \item $f=max(\frac{1}{3}+\frac{2x}{3}, \frac{1}{3}+\frac{2y}{3},
1-\frac{2(x+y)}{3})$. $f$ is minimized when
$x=y=\frac{1}{3}$, which results in $f=\frac{5}{9}$.
\end{enumerate}
So there always exists an axis-parallel rectangle that avoids all the points in the $\epsilon$-net
and contains at least $\frac{5n}{9}$ points.

\end{proof}

\begin{figure}
\begin{center}
 \scalebox{0.25}{\input{r3lb_new.pstex_t}}
\caption{Lower bound for $\epsilon^{\cal{R}}_3$}
\label{fig3RL}
\end{center}
\end{figure}
\vspace{0.5in}
\begin{lem}\label{lemma3RL}
 $\epsilon^{\cal{R}}_3 \geq \frac{2}{5}$ 
\end{lem}

\begin{proof}
Let $P$ be arranged into 20 equal subsets. Each quadrant has a group of five subsets, we will denote these groups as 
$A,B,C,D$ as shown in figure~\ref{fig3RL}.

% Consider $P$ arranged into 20 equal subsets, with five in each quadrant,  as
% shown in the figure. We will number the subsets as $1,2,.....,20$ as shown in
% the figure~\ref{fig3RL}(a). We will call a subset free if no point is chosen
% from it and an
% axis-parallel rectangle big if it contains $\frac{2n}{5}$ points from $P$. We
%claim that whichever three points we choose from $P$, we can have a big
%axis-parallel rectangle that contains only free subsets.

Clearly, all the subsets of one of the groups will be free. Assume that all the subsets
in $A$ are free. The three points in $Q$ can be chosen in two ways.
\begin{itemize}
 \item \emph{Case 1: One point is chosen from each of the other three groups:}
Any set of eight consecutive subsets is contained in a $\frac{2}{5}$-big axis-parallel
rectangle. Since subsets in $A$ are free, all the three subsets near $A$ in $B$ and $D$
cannot be free. This means a point has to be chosen from one of the subsets
$b_1, b_2$ or $b_3$. Now if a
point is chosen from subset $b_1$ then the remaining subsets in $B$ are free and there exists
a $\frac{2}{5}$-big axis-parallel rectangle containing all those subsets and $a_1,a_2,a_3,a_4$ from $A$.
Therefore $b_1$ has to be free. Symmetrically, $d_5$ also has to be free. Clearly two points have to be selected from the subsets $b_2$ and $d_4$
% This
% forces the subsets $b_3$ and $d_3$ to be free 
to avoid the $\frac{2}{5}$-big axis-parallel
rectangle with consecutive subsets from $D,B$ and $A$ ($d_5,a_1,\dots,a_5,b_1,b_2$ and $d_4,d_5,a_1,\dots,a_5,b_1$). If the subset $c_1$ is free then there is a $\frac{2}{5}$-big axis-parallel rectangle containing the free subsets 
$d_5,a_1,a_2,a_3,b_3,b_4,b_5,c_1$.
 Therefore the three points have to be chosen from the subsets
$b_2,c_1,d_4$. Now there exists a $\frac{2}{5}$-big axis-parallel rectangle as shown in figure~\ref{fig3RL}(b) that avoids these subsets.
\item \emph{Case 2: Two points are chosen from a group and the third point is chosen
from the diagonally opposite group:}
Assume that two points are chosen from $B$
and one point from $D$. Clearly, a point has to be chosen from the subset $d_3$ to avoid the $\frac{2}{5}$-big
axis-parallel rectangles with consecutive subsets from $D,A$ and $D,C$ ($d_3,d_4,d_5,a_1,\dots,a_5$ and $c_1,\dots,c_5,d_1,d_2,d_3$).
 The other subsets in $D$ are free. The next two points
chosen have to be from the subsets $b_1$ and $b_5$ to avoid the $\frac{2}{5}$-big axis-parallel
rectangles containing consecutive subsets from $B,C,D$ ($b_5,c_1,\dots,c_5,d_1,d_2$) and $D,A,B$ ($d_4,d_5,a_1,\dots,a_5,b_1$). Now
there exists a $\frac{2}{5}$-big axis-parallel rectangle as shown in figure~\ref{fig3RL}(c)
which avoids the subsets $b_1,b_5,d_3$.
\end{itemize}

\end{proof}

\begin{lem}\label{lemma4RL}
 $\epsilon^{\cal{R}}_4 \geq \frac{3}{10}$
\end{lem}

\begin{proof}
% % Place the $n$ points in ten equal sized groups containing $\frac{n}{10}$
% points each. Eight of the
% % groups are on the vertices of a regular octagon, and two of the groups are
% slightly above the horizontal bisecting line of the
% % octagon, lying on different sides of the vertical bisector as shown in
%figure
%~\ref{fig4RL}(a). For the rectangle shown in figure~\ref{fig4RL}(b) to be
% % hit, one of the top three points has to be selected. The choice of points is
% shown in (c) and (f). The third case is symmetric to
% % (c).
% % 
% % In (c), six axis-parallel rectangles are shown which have the property that
%no
% three have a common intersection. Therefore,
% % a point from the set $P$ can hit at most two of these rectangles. Since
%three
% points can be chosen, each of these must hit two
% % distinct rectangles. So either the point set shown in (d) must be chosen, or
%a
% point from each of the grey regions shown in (e).
% % In (d) a rectangle avoiding the chosen points and containing three of the
% groups can be drawn as shown. In (e), since only one
% % point is to be chosen from each grey region, one of the two dotted
%rectangles
% will not be hit.
% % 
% % For the second case as shown in (f), again six rectangles are drawn with the
% same property as before. (g) and (h) show the two
% % ways of selecting points. In each case, as shown in (g) and (h),
% axis-rectangles containing three of the groups and avoiding the
% % four points in the net can be drawn.
% % 
% % Therefore we can always draw an axis-parallel rectangle avoiding the points
%in
% the net and containing $\frac{3n}{10}$ points.
Let $P$ be arranged as ten equal subsets as shown in figure~\ref{fig4RL}. The eight
subsets in between the topmost and bottommost subsets are arranged as two
quadrilaterals. Let $Q$ be an $\epsilon$-net of size four. We
claim that there always exists a $\frac{3}{10}$-big
axis-parallel rectangle that avoids $Q$. 

It is easy to see that there can be at most two free subsets in each of the
quadrilaterals since there is a $\frac{3}{10}$-big axis-parallel rectangle that contains
any three free subsets from a quadrilateral. Hence all the four points of $Q$ have to be chosen from the  subsets forming the two quadrilaterals(two points from each quadrilateral). Therefore the top and bottom
subsets are free. Both these subsets form a $\frac{3}{10}$-big axis-parallel
rectangle with the black subsets in the quadrilateral(see figure~\ref{fig4RL}). Therefore two of the points in $Q$ have to be chosen from the black subsets. The
two gray subsets in the lower(resp. upper) half form a $\frac{3}{10}$-big axis-parallel
rectangle with the bottommost(resp. topmost) subset. Hence the remaining two points in $Q$ have to be chosen from the gray subsets(one point from the top gray subsets and one from the bottom gray subsets). Thus one gray subset in the
lower half is free which forms a $\frac{3}{10}$-big axis-parallel rectangle that avoids $Q$ as shown
in figure~\ref{fig4RL}.

\end{proof}
\begin{figure}
\begin{minipage}[b]{160pt}

\begin{center}
 \includegraphics[scale=0.5]{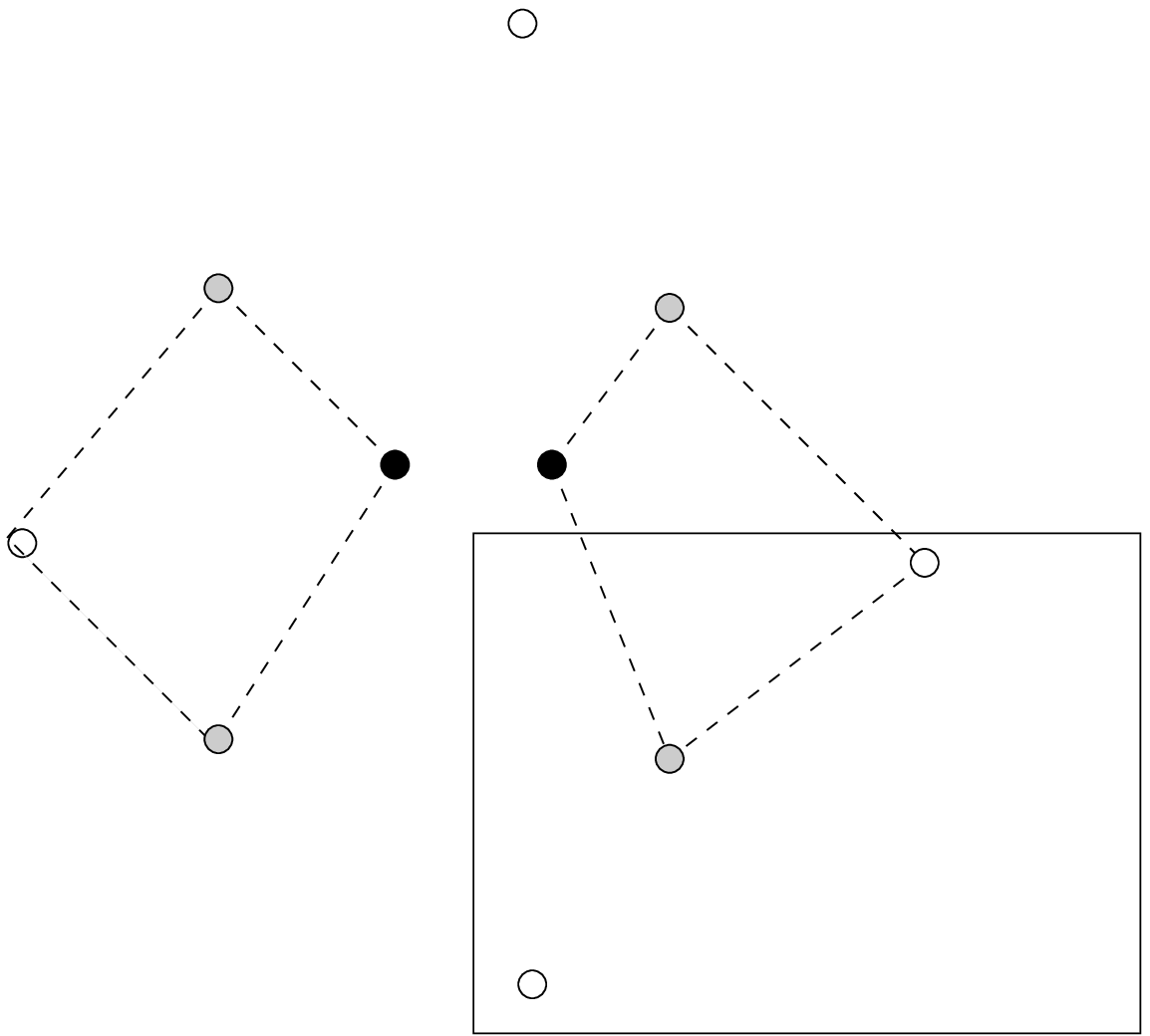}
\caption{Lower bound for $\epsilon^{\cal R}_4$}
\label{fig4RL}
\end{center}

\end{minipage}
\begin{minipage}[b]{200pt}
 \begin{center}
 \includegraphics[scale=0.5]{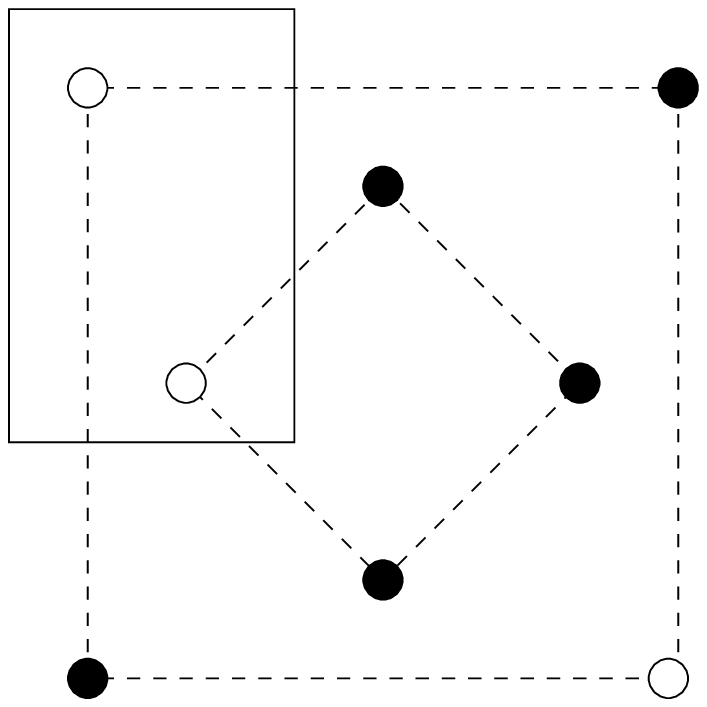}
\caption{Lower bound for $\epsilon^{\cal R}_5$}
\label{fig5lb}
\end{center}
\end{minipage}

\end{figure}

\begin{lem}
 $\epsilon^{\cal{R}}_5 \geq \frac{1}{4}$ 
\end{lem}
\begin{proof}
Let $P$ be
arranged as eight equal subsets, four in the outer layer and four in the
inner layer (see figure~\ref{fig5lb}). Let $Q$ be an $\epsilon$-net of size five. We claim that there always exists a $\frac{1}{4}$-big axis-parallel rectangle that avoids $Q$.  

It is clear that we have to choose a point each from at least three of the inner
layer subsets, otherwise there can be an axis-parallel rectangle that contains
all the points of two free subsets. Also, a point each has to be chosen from
at least two of the outer subsets to avoid the axis-parallel rectangles
that contains all the points of two consecutive subsets in the outer layer. 
Moreover, these two points are to be chosen from diagonally opposite subsets. Let the $\epsilon$-net points be chosen from the five black subsets as shown in figure~\ref{fig5lb}. Now there exists a $\frac{1}{4}$-big axis-parallel rectangle containing the free subset in the
inner layer and the nearest free subset in the outer layer.

\end{proof}

\begin{lem}
$\epsilon_6^{\mathcal{R}} \ge \frac{1}{5}$; $\epsilon_7^{\mathcal{R}} \ge
\frac{5}{29}$; $\epsilon_8^{\mathcal{R}} \ge \frac{2}{13}$;
$\epsilon_9^{\mathcal{R}} \ge \frac{3}{22}$; $\epsilon_{10}^{\mathcal{R}} \ge
\frac{1}{8}$
\end{lem}

\begin{proof}
 The results follow from theorem~\ref{genlower} with $j=k=3$ for
$\epsilon_6^{\mathcal{R}}$;
$j=2, k=5$ for $\epsilon_7^{\mathcal{R}}$; $j=3, k=5$ for
$\epsilon_8^{\mathcal{R}}$; $j=4, k=5$ for $\epsilon_9^{\mathcal{R}}$;
$j=5, k=5$ for $\epsilon_{10}^{\mathcal{R}}$. 
\end{proof}

\begin{lem}
 $\epsilon_i^{\mathcal{R}} \ge \frac{10}{9i}$, for $i \ge 2$
\end{lem}
\begin{proof}
We use mathematical induction to prove this.
The lemma holds for $i=2$ and $i=3$ since $\epsilon_2^{\mathcal{R}} \ge
\frac{5}{9} \ge \frac{10}{9 \times 2}$ (by lemma~\ref{lemma2RL}) and
$\epsilon_3^{\mathcal{R}} \ge \frac{2}{5} \ge \frac{10}{9 \times 3}$ (by lemma~\ref{lemma3RL}). Now assume the lemma holds for $i=k$,
i.e., $\epsilon_k^{\mathcal{R}} \ge \frac{10}{9k}$ or $\frac{9}{10}k \ge
\frac{1}{\epsilon_k^{\mathcal{R}}}$.
Thus, $\frac{9}{10}k + \frac{9}{5} \ge
\frac{1}{\epsilon_k^{\mathcal{R}}} + \frac{1}{\epsilon_2^{\mathcal{R}}}$, 
i.e., $\frac{9}{10} (k+2) \ge \frac{\epsilon_k^{\mathcal{R}} +
\epsilon_2^{\mathcal{R}}}{\epsilon_k^{\mathcal{R}} \epsilon_2^{\mathcal{R}}}$.
Therefore, $\epsilon_{k+2}^{\mathcal{R}} \ge \frac{\epsilon_k^{\mathcal{R}}
\epsilon_2^{\mathcal{R}}}{\epsilon_k^{\mathcal{R}} + \epsilon_2^{\mathcal{R}}} 
\ge \frac{10}{9(k+2)}$ (from theorem~\ref{genlower}).
Hence the result follows for all values of $i$.

\end{proof}

\begin{table}[t]\label{rectable}
\begin{centering}

\begin{tabular}{|c|c|c|c|c|c|c|c|c|c|c|}
\hline
&$\epsilon_{1}$&$\epsilon_{2}$&$\epsilon_{3}$&$\epsilon_{4}$&$\epsilon_{5}
$&$\epsilon_{6}$&$\epsilon_{7}$&$\epsilon_{8}$&$\epsilon_{9}$&$\epsilon_{10}$\\
\hline
LB&3/4&5/9&2/5&3/10&1/4&1/5&5/29&2/13&3/22&1/8\\
\hline
UB&3/4&5/8&9/16&1/2&15/32&15/32&3/7&2/5&5/13&3/8\\
\hline
\end{tabular}
\caption{Summary of lower and upper bounds for $\epsilon_{i}^{\mathcal{R}}$.
}
\end{centering}
\end{table} 

\noindent The summary of lower and upper bounds for $\epsilon_{i}^{\mathcal{R}}$ is given in table 1.

\section{Half spaces}\label{hp}

\noindent Let $P$ be a set of $n$ points in $\mathbb{R}^2$ and $\mathcal{H}$ be the family of halfspaces. In this section, we prove bounds on $\epsilon_i^\mathcal{H}$.

\begin{lem} \label{hslemma}
  $\epsilon_i^\mathcal{H}\leq \frac{2}{i+1}$
\end{lem}
\begin{proof}
 The proof of this lemma is similar to that of Theorem 4.1 in \cite{KPW92}. For $i=2$, \cite{KPW92} gives a constructive proof to show that $\epsilon_2^\mathcal{H}\leq \frac{2}{3}$. For $i \geq 3$, \cite{KPW92} gives an existential proof showing that a strong $\epsilon$-net of size $\lceil \frac{2}{\epsilon} \rceil - 1$ exists. Based on this proof, we give an explicit construction for $\frac{2}{i+1}$-net of size $i$.
% Our argument is an $\epsilon$-net construction based on the results in~\cite{KPW92}. In~\cite{KPW92}, an existential argument is given to show that a strong $\epsilon$-net for halfspaces of size $\lceil \frac{2}{\epsilon} \rceil - 1$ exists. We give a constructive proof to show that $\epsilon_i^{\mathcal{H}} \le \frac{2}{i+1}$.

Let $c_1,c_2,\dots,c_{m}$ be the vertices of the convex hull of $P$, in clockwise order. Let $N$ represent the $\epsilon$-net. We will denote the chosen $\epsilon$-net points as $n_1,n_2,\dots$. Initially $N = \{ c_1 \}$ and $x = c_1$.

Let $c_z$ be the first vertex encountered in the clockwise direction such that the halfspace formed by $\overline{xc_{z+1}}$ and containing $c_z$ has $>\frac{2}{i+1} n$ points. Add $c_z$ to $N$ and set $x = c_z$. Repeat this procedure till we have considered all vertices of the convex hull i.e, we have added the point $n_{k+1}$ to $N$ such that $n_{k+1}$ is $c_1$ or the halfspace $\overline{n_{k}n_{k+1}}$ contains $c_1$. In the latter case, there are two possibilities:
\begin{itemize}
 \item $n_{k+1}$ = $n_2$: Now $n_1$ is redundant and is removed from $N$. Rename $n_k$ as $n_1$.
\item  $n_{k+1}$ lies between $n_1$ and $n_2$ : $n_{k+1} $ is a redundant point and is removed from $N$.
\end{itemize}

Clearly any halfspace that avoids $N$ has at most $\frac{2}{i+1} n$ points. Hence $N$ is a $\frac{2}{i+1}$-net. We claim that $N$ has at most $i$ points. 

Let $N$ have $l$ points and let $H_j$ be the halfspace formed by the line $\overline{n_{j-1} n_{j+1}}$ and containing $n_j$. Let $h_j=\vert H_j \cap P \vert$. From the above construction, we know that $h_j> \frac{2n}{i+1}$ for all $j$, $1\leq j \leq l$. Therefore $\sum\limits_{j=1}^{l} h_j > \frac{2l}{i+1} n$. Since any point in $P$ is present in at most two $H_j$'s, $\sum\limits_{j=1}^{l} h_j \le 2n$. Combining the above two inequalities, we have $l < i+1$ and the result follows.

% Suppose that $|N| = i+1$. Let the elements of $N$ be denoted as $n_1, n_2, ... n_{i+1}$. Le We have $h_j > \frac{2}{i+1} n$ for $j=2,... i+1$ in accordance with the algorithm. So $\sum_{j=2}^{i+1} h_j > \frac{2i}{i+1} n$. Since a point can be counted in at most two $h_j$'s, $\sum_{j=1}^{i+1} h_j \le 2n$. This gives $h_1 < \frac{2}{i+1} n$. Therefore $N\setminus \{ n_1 \}$ is also a valid $\frac{2}{i+1}$-net.
% 

\end{proof}
\begin{lem}\cite{KPW92}
\begin{equation*}
 \noindent \epsilon_i^{\mathcal{H}}\ge
\begin{cases}
  \frac{2}{i+1} &\text{for $i$,odd}\\
 \frac{2}{i+2} &\text{for $i$,even}
\end{cases}
\end{equation*}

\end{lem}
\begin{proof}

The result follows from the lower bounds given in~\cite{KPW92}. For odd values of $i$, consider a point set divided into $\frac{i+1}{2}$ equal subsets, as described in~\cite{KPW92}. For any $\epsilon$-net $N$ of size $i$, there exists a halfspace that avoids $N$ and contains all the $\frac{2}{i+1} n$ points from one subset. For even values, the lower bound for $\epsilon_i^{\mathcal{H}}$ follows from that for $\epsilon_{i+1}^{\mathcal{H}}$.
\end{proof}

\begin{figure}
\begin{centering}
\includegraphics[scale=0.25]{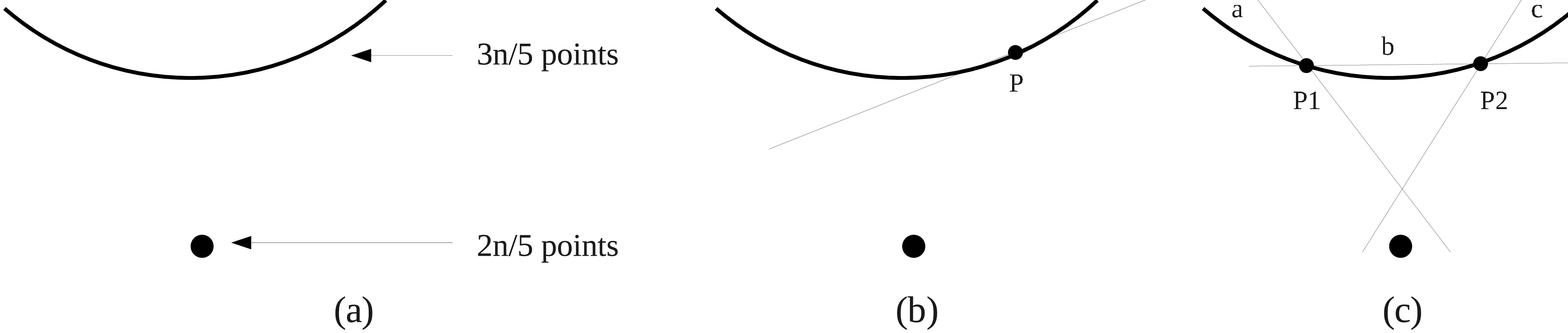}
\caption{Lower bound for $\epsilon^{\cal{H}}_2$}
\label{fig2HL}
\end{centering}
\end{figure}
\begin{lem}\label{lemma2HL}
$\epsilon_2^{\mathcal{H}} \geq \frac{3}{5}$
\end{lem}

\begin{proof}Let $P$ be a set of $n$ points arranged as a subset $S$ of $\frac{3n}{5}$ points uniformly placed on the arc of a circle with angle less than $\pi$ and a subset $T$ of $\frac{2n}{5}$ points placed close together, at a sufficiently large distance from $S$ so that the tangent to the arc at any point $p \in S$ has $S \backslash \{p\}$ and $T$ on different sides of it (see figure~\ref{fig2HL}(a)).

If only one point in the $\epsilon$-net is chosen from $S$, the tangent to the arc passing through it results in a halfspace
containing $\frac{3n}{5}-1$ points (figure~\ref{fig2HL}(b)). So assume that both the $\epsilon$-net points are chosen from $S$. Let $p_1$ and $p_2$ be the $\epsilon_2$-net points chosen from $S$. The points in $S$ are divided into three arcs by $p_1$ and $p_2$. One of the arcs contain at least $\frac{n}{5}$ points. There exists a halfspace avoiding $p_1$ and $p_2$ that contains all the points of this arc and all  the points of $T$(see figure~\ref{fig2HL}(c)). 
\end{proof}

\section{Disks}\label{disk}
\noindent In this section, we consider the family of disks. We show an upper bound for $ \epsilon_2^{\mathcal{D}}$. The proof is motivated by the construction for $\epsilon_2^{\mathcal{H}}$ given in \cite{KPW92}.
\begin{thm}
$\epsilon_2^{\mathcal{D}}\leq{ 2\over3}$
\end{thm}
\begin{proof}
\begin{figure}
\begin{centering}
\includegraphics[scale=0.4]{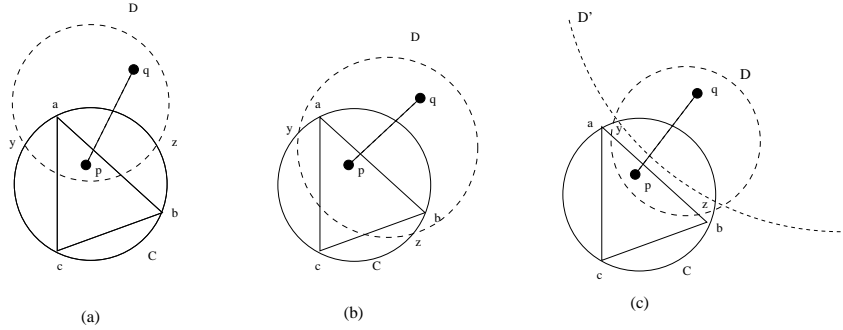}
\caption{Upper bound for $\epsilon_2^{\mathcal{D}}$}
\label{diskfig}
\end{centering}
\end{figure}
Let $P$ be a set of $n$ points and $p$ be the centerpoint of $P$. Let $T$ be the Delaunay triangulation of $P$. Let $a,b,c$ be points in $P$ such that they are the vertices of the delaunay  triangle in $T$ that contains $p$. For all $x\in P\setminus\{a,b,c\}$, consider the line segment connecting $x$ and $p$. Each of these $n-3$ line segments intersect one of the edges $ab,bc$ or $ac$ and at least one of them, say $ab$, has at least one-third of these line segments intersecting it. We claim that $\{a,b\}$ is a ${2\over3}$-net.

 Let $P_1\subset P$ be the set of points $q$ such that the line segment $pq$ intersects the edge $ab$. Consider the circumcircle $C$ of triangle $abc$ and let $C_{ab}$ be the arc between $a$ and $b$ that is intersected by the line segment $pq$, for any $q \in P_1$. By the definition of delaunay triangulation, $C$ does not contain any point of $P$. Any disk $D$ that avoids $\{a,b\}$ and contains more than $2n\over3$ points of $P$ contains the centerpoint $p$ and at least one point $q\in P_1$, since $\vert P_1 \vert \geq \frac{n-3}{3}$. Therefore $D$ intersects $C$ at two points $y,z$. If exactly one of these points lie on the arc $C_{ab}$ then $D$ contains either $a$ or $b$ (see figure~\ref{diskfig}(a)). Similarly if both $y$ and $z$ are outside the arc $C_{ab}$, then $D$ contains both the points $a$ and $b$ (see figure~\ref{diskfig}(b)). Therefore assume that both $y$ and $z$ lie on the arc $C_{ab}$. In this case, $D$ contains only points from $P_1$. Then we can have another disk $D'$ such that $(D\cap P) \subset (D'\cap P$) and $p \notin D'$ (see figure~\ref{diskfig}(c)). This implies that $D'$ contains more than $\frac{2n}{3}$ points and does not contain the centerpoint, a contradiction. 

\end{proof}

The lower bounds $\epsilon_i^{\mathcal{D}} \ge \frac{2}{i+1}$ for odd values of $i$ and $\epsilon_i^{\mathcal{D}} \ge \frac{2}{i+2}$ for even values of $i$ follow from the lower bound results for halfspaces.

\section{Lower bounds for weak epsilon nets}\label{lbweak}
 
In this section, we give general lower bounds for $\epsilon_i$ for weak $\epsilon$-nets. These bounds are improvements over the general lower bound  of $\frac{1}{i+1}$ given in \cite{AAH09}.

\noindent Let $\mathcal{C}$ represent the family of convex objects.

%%%%%%%%%%%%%%%%%%%%%%% Weak nets for convex objects
\begin{lem}
 $\epsilon_4^{\mathcal{C}} \ge \frac{2}{7}$; $\epsilon_5^{\mathcal{C}} \ge \frac{20}{79}$;
$\epsilon_6^{\mathcal{C}} \ge \frac{5}{22}$; $\epsilon_7^{\mathcal{C}} \ge \frac{10}{57}$; $\epsilon_8^{\mathcal{C}} \ge \frac{20}{123}$; $\epsilon_9^{\mathcal{C}} \ge \frac{5}{33}$; $\epsilon_{10}^{\mathcal{C}} \ge \frac{10}{79}$
\end{lem}

\begin{proof}
The results follow from theorem \ref{genlower} by applying the base cases $\epsilon_2^{\mathcal{C}} \ge \frac{4}{7}, \epsilon_3^{\mathcal{C}} \ge \frac{5}{11}$ (as shown in~\cite{MR09}). Use $j=k=2$ for $\epsilon_4^{\mathcal{C}}$; $j=2, k=3$ for $\epsilon_5^{\mathcal{C}}$;
$j=k=3$ for $\epsilon_6^{\mathcal{C}}$; $j=2, k=5$ for $\epsilon_7^{\mathcal{C}}$;
$j=3, k=5$ for $\epsilon_8^{\mathcal{C}}$; $j=3, k=6$ for $\epsilon_9^{\mathcal{C}}$;
$j=2, k=8$ for $\epsilon_{10}^{\mathcal{C}}$.
\end{proof}
\begin{lem} 
 $\epsilon_i^{\mathcal{C}} \ge \frac{8}{7i}$, for $i \ge 2$
\end{lem}
\begin{proof}
We prove the result using mathematical induction.
The result holds for $i=2$ and $i=3$ since $\epsilon_2^{\mathcal{C}} \ge \frac{4}{7} \ge \frac{8}{7 \times 2}$ and $\epsilon_3^{\mathcal{C}} \ge \frac{5}{11} \ge \frac{8}{7 \times 3}$ (both lower bounds shown in~\cite{MR09}).
Assume the result holds for $i=k$
i.e., $\epsilon_k^{\mathcal{C}} \ge \frac{8}{7k}$ or $ \frac{7}{8}k \ge \frac{1}{\epsilon_k^{\mathcal{C}}}$.
Thus, $\frac{7}{8}k + \frac{7}{4} \ge \frac{1}{\epsilon_k^{\mathcal{C}}} + \frac{1}{\epsilon_2^{\mathcal{C}}}$ (since $\epsilon_2^{\mathcal{C}} \ge \frac{4}{7}$ as shown in~\cite{MR09})
i.e., $\frac{7}{4} (k+2) \ge \frac{\epsilon_k^{\mathcal{C}} + \epsilon_2^{\mathcal{C}}}{\epsilon_k^{\mathcal{C}} \epsilon_2^{\mathcal{C}}}$.
Therefore, $\epsilon_{k+2}^{\mathcal{C}} \ge \frac{\epsilon_k^{\mathcal{C}} \epsilon_2^{\mathcal{C}}}{\epsilon_k^{\mathcal{C}} + \epsilon_2^{\mathcal{C}}}  \ge \frac{8}{7(k+2)}$ (from theorem \ref{genlower}).

\end{proof}

%%%%%%%%%%%%%%%%%%%%%%% Weak nets for disks

\begin{lem}\label{disk3lower}
 $\epsilon_3^{\mathcal{D}} \ge \frac{1}{3}$
\end{lem}
\begin{proof}

\begin{figure}
\begin{centering}
\includegraphics[scale=0.25]{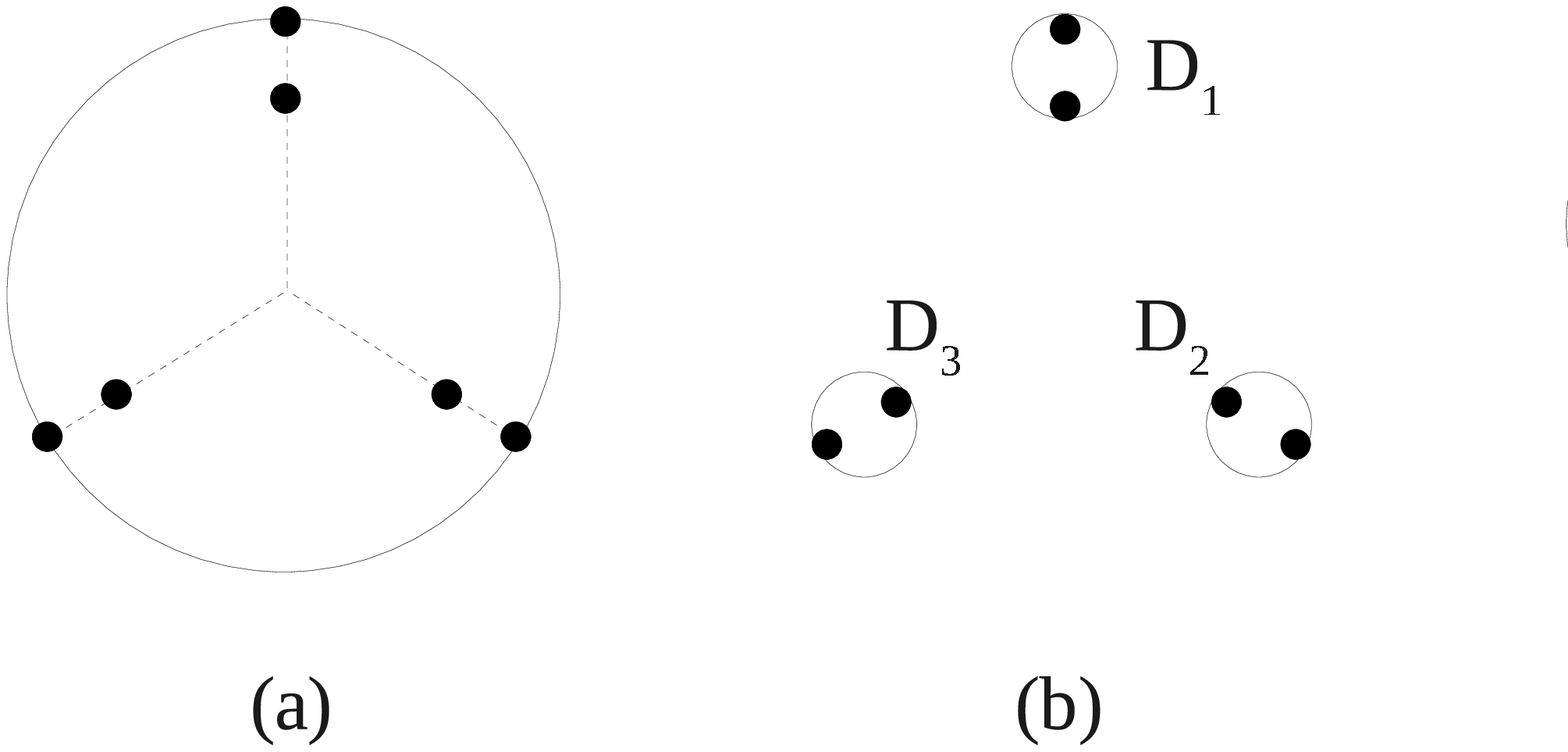}
\caption{Lower bound for $\epsilon^{\cal{D}}_3$}
\label{fig3DL}
\end{centering}
\end{figure}
 Let $P$ be a set of $n$ points divided into six subsets of $\frac{n}{6}$ points each and arranged equidistant on two concentric circles as shown in figure \ref{fig3DL}(a).Let $Q$ be a weak $\epsilon$-net of size three. We claim that there always exists a disk that avoids $Q$ and contains $\frac{n}{3}$ points from $P$.

 The three disks, $D_1$, $D_2$ and $D_3$ shown in figure \ref{fig3DL}(b) contain $\frac{n}{3}$ points each. So each of these disks must contain one point from $Q$. Disk $D'_{1}$ and halfspace $H_{12}$ as shown in figure \ref{fig3DL}(c) also contain $\frac{n}{3}$ points each. Therefore, the point in $Q$ from disk $D_1$ must also be contained in disk $D'_1$. Similarly, the point from either $D_1$ or $D_2$ must be contained in halfspace $H_{12}$. Without loss of generality, let the point of $Q$ contained in $H_{12}$ be from $D_1$. A point belonging to $D_1 \cap D'_1 \cap H_{12}$ belongs to the topmost group as can be seen in figure \ref{fig3DL}(c). Now halfspace $H_{23}$, shown in figure \ref{fig3DL}(d) must contain a point belonging to either $D_2$ or $D_3$. Let this point lie in $D_3$. Then there exists a disk that avoids $Q$ and contains $\frac{n}{3}$ points as shown in figure \ref{fig3DL}(d).
\end{proof}

\begin{lem}
 $\epsilon_i^{\mathcal{D}} \ge \frac{1}{i}$, for $i \ge 2$
\end{lem}

\begin{proof}
We use mathematical induction to prove the result.
The result holds for $i=2$(by \cite{AAH09}) and $i=3$(by lemma \ref{disk3lower}).
Assume the result holds for $i=k$
i.e., $\epsilon_k^{\mathcal{D}} \ge \frac{1}{k}$ or $k \ge \frac{1}{\epsilon_k^{\mathcal{D}}}$.
Thus,  $k+2 \ge \frac{1}{\epsilon_k^{\mathcal{D}}} + \frac{1}{\epsilon_2^{\mathcal{D}}}$ (since $\epsilon_2^{\mathcal{D}} \ge \frac{1}{2}$ as shown in~\cite{AAH09}) 
i.e.,  $k+2 \ge \frac{\epsilon_k^{\mathcal{D}} + \epsilon_2^{\mathcal{D}}}{\epsilon_k^{\mathcal{D}} \epsilon_2^{\mathcal{D}}}$. Thus 
$\epsilon_{k+2}^{\mathcal{D}} \ge \frac{\epsilon_k^{\mathcal{D}}\epsilon_2^{\mathcal{D}}}{\epsilon_k^{\mathcal{D}} + \epsilon_2^{\mathcal{D}}} \ge \frac{1}{k+2}$.
\end{proof}

\begin{lem}\label{circ}
When a circle is divided into $i$ equal sectors, $i\geq4$, there exist $i$ axis-parallel rectangles such that each of them contains the corresponding arc from each sector and none of them intersect.
\end{lem}
\begin{proof}

We will first prove for $i=4$.

Assume that the centre of the circle is at the origin and the radius of the circle is $r$. Then the two lines that divide the circle into four equal sectors will pass through (0,0) and will intersect the circle in four different points. Let these points of intersection be $(a,b),(c,-d), (-a,-b),(-c,d) $ where $a,b,c,d \geq 0$.
\begin{figure}
\begin{centering}
\includegraphics[height=2in]{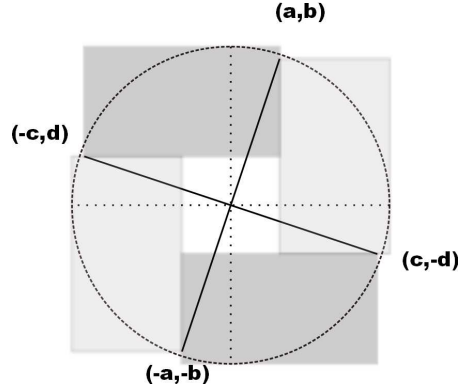}
\caption{To prove Lemma \ref{circ}}
\label{figcir}
\end{centering}

\end{figure}

Without loss of generality, assume that $d\leq b$. Now the four rectangles with the following set of vertices will contain the four arcs (See Figure \ref{figcir}).
\begin{itemize}
\item$(-c,d),(-c,r),(a,r),(a,d)$
\item$(a,b),(r,b),(r,-d),(a,-d)$
\item$(c,-d),(c,r),(-a,r),(-a,-d)$
\item$(-a,-b),(r,-b),(r,d),(-a,d)$
\end{itemize}
Clearly none of these rectangles intersect. 

Now assume $i > 4$. When a circle is divided into $i$ sectors, two rectangles that contain two consecutive sectors are contained in two separate bigger rectangles as seen in the previous paragraph. Since the bigger rectangles are already proved to be non-intersecting the smaller rectangles that are contained in them also do not intersect.
\end{proof}
\begin{thm}
$\epsilon_{i}^{\mathcal{R}} \geq {1\over i}$, for $i\geq 4$
\end{thm}
\begin{proof}
Let $P={p_1, p_2,...., p_n}$ be a point set of $n$ points arranged along the boundary of a circle, where $n=ik+1$ for some arbitrary $k$. Consider groups of points, $S_i$, $1\leq i \leq n$, of size $k$ in which each group consists of $k$ consecutive points starting with point $p_i$. Now, by Lemma~\ref{circ}, there exists a family of axis-parallel rectangles  $R=\{R_{j},1\leq j \leq n\}$ such that each $R_j$ contains all the $k$ points of $S_j$ and $R_{j}$ and $R_{j+k}$ do not have a common intersection. Therefore any point $p$ in the epsilon net can cover at most $k$ rectangles in $R$, so $i$ points can cover at most $ik$ rectangles. Thus, for any $\epsilon$-net $Q$ of size $i$, there exists an axis-parallel rectangle in $R$ that avoids $Q$. Thus for large $k$, $\epsilon_{i}^{\mathcal{R}} \geq {1\over i}$.
\end{proof}

\section*{Conclusions and Open Questions}

In this paper, we have shown lower and upper bounds for $\epsilon_i^{\mathcal S}$ where $\mathcal{S}$ is the family of axis-parallel rectangles, halfspace and disks. A summary of the bounds for $i \leq 3$ are given in Table 2. An interesting open question is to find the exact value of $\epsilon_i^{\mathcal{S}}$ for small values of $i$. 
\noindent
\begin{table}[t] \label{sumtable}
\begin{centering}

\begin{tabular}{|c|c|c|c|c|c|c|}
\hline
&\multicolumn{2}{|c|}{Rectangles}&\multicolumn{2}{|c|}{Half Spaces}&\multicolumn{2}{|c|}{Disks}\\
\hline
&LB&UB&LB&UB&LB&UB\\
\hline
$\epsilon_1$&\multicolumn{2}{|c|}{$3/4$}&\multicolumn{2}{|c|}{1}&\multicolumn{2}{|c|}{1}\\
\hline
$\epsilon_2$&$5/9$&$5/8$&$3/5$&$2/3$&$3/5$&$2/3$\\
\hline
$\epsilon_3$&$2/5$&$9/16$&\multicolumn{2}{|c|}{$1/2$}&$1/2$&$2/3$\\
\hline
\end{tabular}
\caption{Summary of lower and upper bounds for $\epsilon_{i}$.
}
\end{centering}
\end{table}
\section*{Acknowledgements}
\noindent We would like to thank Janardhan Kulkarni for helpful discussions.

% \bibliography{myref}
\end{document}